\documentclass[11pt,onecolumn,final]{article} 
\usepackage{a4}
\usepackage{fullpage}
\usepackage[latin1]{inputenc}
\usepackage{a4}
\usepackage{amsfonts}
\usepackage{amsmath}
\usepackage{amssymb}
\usepackage{amsthm}
\usepackage{graphicx}
\usepackage[ruled,longend]{algorithm}
\usepackage{subfigure}
\usepackage{stmaryrd}
\usepackage{fancybox}
\usepackage{url}
\usepackage{cite}
\usepackage{multirow}
\usepackage{hyperref}
\usepackage{enumerate}
\usepackage{color}

\newcommand{\SA}[0]{\mathsf{SA}} 
\newcommand{\ST}[0]{\mathcal{ST}} 
\newcommand{\SF}[1]{T^{#1}} 
\newcommand{\pred}{\textsc{Pred}} 

\theoremstyle{plain}
\newtheorem{definition}{Definition}
\newtheorem{lemma}[definition]{Lemma}
\newtheorem{prop}[definition]{Proposition}
\newtheorem{theorem}[definition]{Theorem}
\newtheorem{observation}[definition]{Observation}
\newtheorem{corollary}[definition]{Corollary}
\theoremstyle{remark}

\definecolor{myRed}{rgb}{0.9,0.3,0.1}

\definecolor{myYellow}{rgb}{0.99,0.99,0.05}

\title{Alphabet-Dependent String Searching with \\Wexponential Search Trees}

\author{Johannes Fischer and Pawe{\l} Gawrychowski}

\begin{document}
\maketitle

\begin{abstract}
It is widely assumed that $O(m+\lg \sigma)$ is the best one can do for finding a pattern of length $m$ in a compacted trie storing strings over an alphabet of size $\sigma$, if one insists on linear-size data structures and deterministic worst-case running times [Cole et al., ICALP'06]. In this article, we first show that a rather straightforward combination of well-known ideas yields $O(m+\lg\lg \sigma)$ deterministic worst-case searching time for static tries.

Then we move on to dynamic tries, where we achieve a worst-case bound of $O(m+\frac{\lg^{2}\lg\sigma}{\lg\lg\lg\sigma})$ per query or update, which should again be compared to the previously known $O(m+\lg\sigma)$ deterministic worst-case bounds [Cole et al., ICALP'06], and to the alphabet \emph{in}dependent $O(m+\sqrt{\lg n/\lg\lg n})$ deterministic worst-case bounds [Andersson and Thorup, SODA'01], where $n$ is the number of nodes in the trie. The basis of our update procedure is a weighted variant of exponential search trees which, while simple, might be of independent interest.

As one particular application, the above bounds (static and dynamic) apply to suffix trees. There, an update corresponds to pre- or appending a letter to the text, and an additional goal is to do the updates quicker than rematching entire suffixes. We show how to do this in $O(\lg\lg n + \frac{\lg^{2}\lg\sigma}{\lg\lg\lg\sigma})$ time, which improves the previously known $O(\lg n)$ bound [Amir et al., SPIRE'05].
\end{abstract}

\section{Introduction}
Text indexing is a fundamental problem in computer science. It requires storing a text of length $n$, composed of letters from an alphabet of size $\sigma$, such that subsequent pattern matching queries can be answered quickly. Typical such pattern matching queries are (1) \emph{existential} queries (deciding whether or not the pattern occurs in the text), (2) \emph{counting} queries (determining the number of occurrences), and (3) \emph{enumeration} queries (listing all positions where the pattern occurs). The text is either static, or can be modified by appending new letters.

Well-known static text indexes are \emph{suffix trees} \cite{mccreight76space} and \emph{suffix arrays} \cite{manber93suffix}. The former admit, in their plain form, $O(m\lg\sigma)$ existential and counting queries for a pattern of length $m$, while the latter achieve $O(m+\lg n)$ time with the help of two additional arrays of size $n$ storing the lengths of longest common prefixes needed during the binary search. With perfect hashing \cite{fredman84storing}, suffix trees achieve $O(m)$ time, but then the $O(n)$ preprocessing time is only in \emph{expectation}. Enjoying the best of both worlds, the \emph{suffix tray} of Cole et al. \cite{cole06suffix} achieves $O(m+\lg \sigma)$ searching time, with a linear space data structure that can be constructed in $O(n)$ deterministic time for a static text. In the dynamic setting, suffix trees can be updated in amortized expected constant time, where the amortization comes from the need to locate the node which should be updated, and expectation from the hashing used to store outgoing edges. If we insist on getting worst-case time bounds, a recent result of Kopelowitz~\cite{Kopelot12indexing} allows updates in $O(\lg\lg n +\lg\lg\sigma)$ worst-case (but still expected) time. The \emph{suffix trists} of Cole et al.~\cite{cole06suffix} achieve a deterministic bound 
of $O(m+\lg\sigma)$ for searching, with a linear space data structure that can be updated in $O(f(n,\sigma)+\lg\sigma)$ deterministic and worst-case time, where $f(n,\sigma)$ is the time required to locate the edge of the suffix tree which should be split, and the best bound on $f(n,\sigma)$ known so far for the general case is $O(\lg n)$~\cite{Amir05towards}. In all cases, $O(occ)$ additional time is needed for enumerating the $occ$ occurrences.

In all of the tree based structures mentioned in the previous paragraph, the crucial point is how to implement the outgoing edges of the tree such that they can be searched efficiently for a given query character. This is the general setting of \emph{trie search}, and in fact, we can and do formulate our results in terms of tries, and view suffix trees as one particular application. In this setting, it is worth mentioning the result of Andersson and Thorup \cite{andersson07dynamic}, who show how to update or search the trie in $O(m+\sqrt{\lg n/\lg\lg n})$ deterministic worst-case time, for $n$ being the number of stored strings. In this article, however, we focus on alphabet-\emph{dependent} running times, as outlined next.

\subsection{Our Result and Outline}
Our main results are summarized in the following theorems. We work in the standard word RAM model, where it is assumed that the memory consists of $\Omega(\lg n)$-bit cells that can be manipulated using standard C-like operations in $O(1)$ time.

In tries, we support stronger forms of existential queries, namely prefix queries (deciding whether or not the search pattern is a prefix of one of the stored strings), and predecessor queries (returning the largest string stored that is no less than the query pattern):

\begin{theorem}
  \label{thm:main}
  A compacted trie storing $n$ static strings over an alphabet of size $\sigma$ can be stored in $O(n)$ space (in addition to the stored strings themselves) such that subsequent prefix- or predecessor queries can be answered in $O(m+\lg\lg\sigma)$ deterministic worst-case time, for patterns of length $m$. This data structure can be constructed in deterministic $O(n)$ time, for $\sigma=n^{O(1)}$.
\end{theorem}

We note that while for prefix queries the ``$+\lg\lg\sigma$''-term would in principle not be necessary, for the supported predecessor queries it certainly is. As one application, we mention suffix trees \cite{weiner73linear}:

\begin{corollary}
  \label{cor:main}
  Given a static text of length $n$ over an alphabet of size $\sigma=n^{O(1)}$, we can build in $O(n)$ time a linear-space data structure such that subsequent on-line pattern matching queries can be answered in $O(m+\lg\lg \sigma)$ deterministic time, for patterns of length $m$.
\end{corollary}

In the dynamic setting we get the following result:

\begin{theorem}
\label{thm:main2}
  We can maintain a linear-size structure for a trie storing strings over an alphabet of size $\sigma$ under adding a new string of length $\ell$ in deterministic worst-case time $O(\ell+\frac{\lg^{2}\lg\sigma}{\lg\lg\lg\sigma})$ so that subsequent on-line prefix- or predecessor queries can be answered in $O(m+\frac{\lg^{2}\lg\sigma}{\lg\lg\lg\sigma})$ deterministic time, for patterns of length $m$.
\end{theorem}

This result can be formulated in terms of suffix trees, too:

\begin{corollary}
  \label{cor:main2}
  We can maintain a linear-size structure for a text over an alphabet of size $\sigma=n^{O(1)}$ under prepending a new letter in deterministic worst-case time $O(f(n,\sigma)+\frac{\lg^{2}\lg\sigma}{\lg\lg\lg\sigma})$ so that on-line pattern matching concerning the current text can be answered in $O(m+\frac{\lg^{2}\lg\sigma}{\lg\lg\lg\sigma})$ deterministic worst-case time, for patterns of length $m$. $f(n,\sigma)$ is the cost of updating the suffix tree for a text $n$ over an alphabet of size $\sigma$ after prepending a letter.
\end{corollary}


It was already noted that the currently best bound \cite{Amir05towards} for deterministic worst-case suffix tree oracles is $f(n,\sigma)=\lg n$. Though not being the main goal of this article, in the appendix we show a better suffix tree oracle, giving us truly better total running times:

\begin{theorem}
  \label{thm:main3}
There is a deterministic worst-case suffix tree oracle with $f(n,\sigma)=O(\lg\lg n + \frac{\lg^2\lg\sigma}{\lg\lg\lg\sigma})$.
\end{theorem}

\subsection{Technical Contributions}
Our main technical novelty is a weighted variant of exponential search trees \cite{andersson07dynamic}, which we term \emph{wexponential search trees}. The original exponential search tree achieves $O(\lg\lg n \cdot \frac{\lg\lg u}{\lg\lg\lg u})$ search and update times for $n$ elements over a universe of size $u$. Our weighted variant generalizes this to $O(\lg\frac{\lg W}{\lg w} \cdot \frac{\lg\lg u}{\lg\lg\lg u})$, where $w$ is the weight of the searched element, and $W$ is the sum of all weights. The advantage of this is that in a sequence of $t$ hierachical accesses to such data structures, where the old ``$w$'' is always the new ``$W$'', the sum telescopes to $O(\lg\lg n \cdot \frac{\lg\lg u}{\lg\lg\lg u})$ instead of $O(t\cdot\lg\lg n \cdot \frac{\lg\lg u}{\lg\lg\lg u})$. While this general idea is pervasive in data structure, we are not aware of any previous application in the doubly-logarithmic setting.

\section{Preliminaries}
\label{sect:preliminaries}
This section introduces definitions and known results that form the base of new data structure.

\subsection{Suffix Arrays}
Let $T=t_1\dots t_n$ be a text consisting of $n$ characters drawn from an ordered alphabet $\Sigma$ of size $\sigma = |\Sigma|$. The substring of $T$ ranging from $i$ to $j$ is denoted by $T_{i..j}$, for $1\le i \le j \le n$. The substring $T_{i..n}$ is called the $i$'th \emph{suffix} of $T$ and is denoted by $\SF{i}$.

The \emph{suffix array} $\SA[1,n]$ of $T$ is a permutation of the integers in $[1,n]$ such that $\SF{\SA[i-1]} <_\mathrm{lex} \SF{\SA[i]}$ for all $1 < i \le n$. In other words, $\SA$ describes the lexicographic order of the suffixes. The suffix array can be built in linear time for $\sigma=n^{O(1)}$~\cite{kaerkkaeinen06linear}.

\subsection{Suffix Trees}
The suffix tree $\ST$ is a compacted trie over all the suffixes of $T$. By concatenating all edge labels on a root-to-node path, its nodes spell out substrings of $T$. If we append a unique character $\$\not\in\Sigma$ to $T$, then every leaf spells out exactly one suffix of $T\$$. Indeed, if the leaves are visited in a lexicographically driven depth first traversal of $\ST$, then the suffixes are visited in the order of the suffix \emph{array} $\SA$. This implies that every node $v$ of $\ST$ represents an \emph{interval} $[\ell_v,r_v]$ in $\SA$, whose endpoints $\ell_v$ and $r_v$ can be stored with $v$: if the root-to-$v$ path spells out $\alpha\in\Sigma^\star$, then $v$ stores two integers $\ell_v$ and $r_v$, such that $\SF{\SA[\ell_v]}$ is the lexicographically first suffix starting with $\alpha$, and $\SF{\SA[r_v]}$ is lexicographically last. We call these the \emph{suffix array intervals}.

Given the suffix array, the suffix tree can be constructed in $O(n)$ time by repeatedly inserting the suffixes in lexicographic order, with the help of the so-called \emph{LCP-array} \cite{gusfield97algorithms}.

Given the suffix tree for a text $T$, constructing the suffix tree for a text $aT$, for any letter $a$, requires adding just one additional suffix, which in a compacted trie means splitting at most one edge into two parts and adding a single edge outgoing from the middle node (or one of the already existing nodes, if no new middle node was created). We call this \emph{splitting} an edge. In the dynamic setting we will assume the existence of a \emph{suffix tree oracle} telling us which edge should be split after prepending $a$, and denote the time taken by a single such call by $f(n,\sigma)$. The oracle is assumed to be deterministic, but the bound on $f(n,\sigma)$ can be amortized if we aim at getting amortized time bounds.


\subsection{Suffix Trays}
\label{sect:stray}
The \emph{suffix tray} \cite{cole06suffix} is a blend of the suffix tree and the suffix array. The idea is to discard small subtrees of the suffix tree and only keep the upper part where nodes have sufficiently many leaves below them. Then this upper part can be augmented with ``expensive'' information to enable $O(m)$ pattern searches if the pattern ends in the upper part. For patterns not ending in the upper part, one switches to binary suffix array search as outlined in the introduction. If the binary searched intervals are of size $O(\sigma^{O(1)})$, this latter step costs only $O(m+\lg \sigma)$ time.

\begin{figure}
\centering
\includegraphics[scale=1]{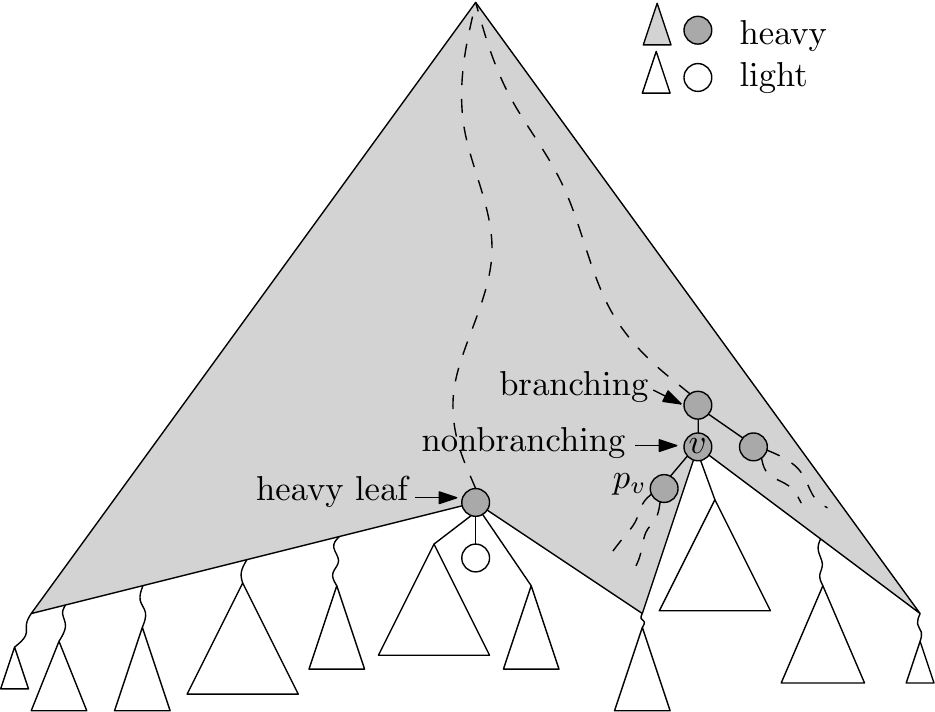}
\caption{Schematic view of a suffix tray.}
\label{fig:stray}
\end{figure}

More precisely, we classify the suffix tree nodes into heavy and light as follows: a node is called \emph{heavy} if it has at least $\sigma$ leaves in its subtree; otherwise it is called \emph{light}. The heavy nodes form a connected subtree (this is the \emph{upper part} aluded to above). Heavy nodes are further subdivided into \emph{branching} and \emph{nonbranching}, depending on whether they have at least two heavy children or just one. See also Fig.~\ref{fig:stray}.

Light children of heavy nodes store their corresponding suffix array interval. The heavy nodes store additional information as described next. Nonbranching heavy nodes $v$ store a special pointer $p_v$ to their only heavy child. Branching heavy nodes $v$ store an array $A_v$ of size $\sigma$, where $A_v[a]$ stores a pointer to the node $w$ such that the first character on the edge $(v,w)$ is $a\in\Sigma$. Since there are only $O(\frac{n}{\sigma})$ branching heavy nodes, these arrays take only $O(n)$ space in total.

Now a search for a pattern $P[1,m]$ proceeds as follows. Assume inductively that we have already matched $P[1,i]$ and are currently at node $v$ (we start at the root with $i=0$). At branching heavy nodes $v$, we can move in $O(1)$ time to the correct child using the array $A_v$. If the current node $v$ is nonbranching internal, we first check if the first character on the edge $(v,p_v)$ equals $P[i+1]$. If so, we proceed with $p_v$. Otherwise, we binary search among the children of $v$ for the character $P[i+1]$. This will bring us to a light child $w$ of $v$, where we binary search $\SA$ within the suffix array interval $[\ell_w, r_w]$. A similar binary search is performed if the array $A_v$ directs us to a light node. Since there are at most two binary searches on arrays of size $O(\sigma)$, the claimed running time of $O(m+\lg\sigma)$ follows.

\subsection{Suffix Trists}

The \emph{suffix trist} \cite{cole06suffix} is a dynamic version of the suffix tray. The idea is, as in the suffix tray, to maintain the suffix tree, and store the outgoing edges in different ways depending on the size of the subtree. We assume that a suffix tree oracle is available, and hence on updating the sets on those edges. For nodes with a large number of leaves in their subtrees, we can afford to store the outgoing edges in a large but quick to both access and update data structure. For nodes with just a few leaves, we switch to a different representation, called Balanced-Indexing-Structure. The main difficulty is that the size of a subtree increases in time, and hence one needs to cleverly move from one representation to another. Even more challenging part of the problem is that we want the bounds to be worst-case, hence the transition between different representations must be done in an incremental fashion.

We do not describe the details of the original solution, as our idea is slightly different. We just state that its time bounds are $O(f(n,\sigma)+\lg\sigma)$ for updates and $O(m+\lg\sigma)$ for queries, both deterministic worst-case.

\subsection{Deterministic Hashing and Predecessor Data Structures}
\label{sect:hashing}
We make use of the following result by Ru\v{z}i\'{c} \cite[Theorem~3]{ruzic08constructing}. Note that even the $O(k\lg k)$ construction time of Hagerup, Miltersen, and Pagh~\cite{hagerup01dictionaries} would be enough for our purposes, but their implementation requires some compile-time constants not known to be efficiently computable, and does not seem to be significantly simpler.

\begin{prop}
  \label{thm:hashing}
  In the word RAM model, a static linear-size dictionary on a set of $k$ keys can be deterministically constructed in time $O(k\lg^2\lg k)$, so that lookups to the dictionary take time $O(1)$.
\end{prop}

Combining Prop.~\ref{thm:hashing} with the classic $x$-fast tries by Willard \cite{willard83loglogarithmic}, we get the following:

\begin{prop}
  \label{thm:predecessor}
In the word RAM model, a static linear-size data structure on a set $S$ of $k$ sorted keys from a universe of size $u$ can be deterministically constructed in time $O(k)$, so that subsequent \emph{predecessor queries} $\pred(x) := \max\{y\le x\mid y\in S\}$ can be answered in $O(\lg\lg u)$ time in the worst-case.
\end{prop}

\begin{proof}
Achieving $O(k\lg u\lg^{2}\lg (k\lg u))$ construction time and $O(k\lg u)$ space is trivial, as constructing the $x$-fast tree reduces to storing just $k\lg u$ elements in a static dictionary. To speed up the precomputation and reduce the space to linear, we choose $\frac{k}{\lg u\lg^{2}\lg k}$ evenly spaced keys, store them in an $x$-fast tree constructed in time $O(\frac{k}{\lg u\lg^{2}\lg k}\lg u\lg^{2}\lg k)=O(k)$, and use binary search to locate the predecessor between two adjacent chosen keys to answer each query.
\end{proof}

We also make use of the result of Beame and Fich \cite[Theorem~3]{beame02predecessor}, who achieved optimal query time possible in the static setting. For the sake of concreteness, we will assume $\epsilon=1$.

\begin{prop}
\label{thm:predecessor2}
In the word RAM model, a static data structure on a set $S$ of $k$ keys from a universe of size $u$ can be deterministically constructed in $O(k^{1+\epsilon})$ time and space, so that subsequent \emph{predecessor queries} $\pred(x) := \max\{y\le x\mid y\in S\}$ can be answered in $O(\frac{\lg\lg u}{\lg\lg\lg u})$ time in the worst-case.
\end{prop}

\begin{proof}
In the original formulation the construction time and space were $O(k^{4})$. We can improve on this with a constant number of indirection steps. More precisely, given a structure with $O(k^{1+\epsilon})$ construction time and space, we can implement a $O(k^{1+\frac{\epsilon}{2}})$ time and space structure by choosing $\sqrt{k}$ evenly spaced keys, storing them in the more expensive structure, and building a separate more expensive structure for each group of keys between two chosen elements.
\end{proof}

The above structure can be dynamized using the method of Andersson and Thorup~\cite{andersson07dynamic}.
\begin{prop}
\label{thm:predecessor3}
In the word RAM model, a dynamic linear-size data structure on a set $S$ of $k$ keys from a universe of size $u$ can be maintained, so that subsequent \emph{predecessor queries} $\pred(x) := \max\{y\le x\mid y\in S\}$ can be answered in $O(\frac{\lg^{2}\lg u}{\lg\lg\lg u})$ time, and new keys can be inserted in $O(\frac{\lg^{2}\lg u}{\lg\lg\lg u})$ time, where both bounds are deterministic  worst-case.
\end{prop}

\section{New Static Data Structure}
In this section we prove Theorem~\ref{thm:main}. We store \emph{edges} of the trie as pairs of the form $(v,a)$, where $v$ is a pointer to the source of the edge, and $a\in\Sigma$ is the first character on the edge from $v$ to its corresponding child $w$. As secondary information we also attach a pointer to $w$ with the pair $(v,a)$. This way, matching a pattern $P[1,m]$ reduces to repeatedly finding correct edges: assuming inductively that we have already matched $P[1,i]$ and are currently at node $v$, we check if the edge $(v,P[i+1])$ exists, and move to that child if this is the case.

We now describe how the edges are stored. A naive storage with the data structures from Sect.~\ref{sect:hashing} would result in superlinear construction time. Hence, we need to introduce several levels of indirection. Like in the suffix tray (Sect.~\ref{sect:stray}), we divide nodes into heavy and light, but this time with parameter $s := \Theta(\lg^2\lg\sigma)$: a node with at least $s$ leaves below it is called \emph{heavy}, otherwise it is called \emph{light}. 

To continue, we classify the edges of heavy nodes into two different types: (heavy,heavy) and (heavy,light). Here, the first component of each tuple refers to the type of the parent node, and the second one to that of the child. For each branching heavy node we store the outgoing edges of type (heavy,heavy) in a dictionary using Prop.~\ref{thm:hashing}, with the key being the first character on the edge. Since there are at most $\frac{n}{s}$ heavy nodes with no heavy children, the total size of all dictionaries cannot exceed the sum of degrees in a tree on $2\frac{n}{s}$ nodes, which
is $O(\frac{n}{s})$. Furthermore,
the number of elements in each dictionary is at most $\sigma$, so constructing all those data structures takes $O(\frac{n}{s} \times \lg^2\lg\sigma) = O(n)$ time. For nonbranching heavy nodes, we only have at most one edge of type (heavy,heavy), so each such node simply stores a special pointer to the only heavy child, which enables us to decide in $O(1)$ time if we need to continue matching there.

Finally, we show how to handle the edges of type (heavy,light). For each (branching or not) heavy node we store all such outgoing edges using the data structure from Prop.~\ref{thm:predecessor}. Using the structure, we can locate the light child we should descend to in $O(\lg\lg\sigma)$ time. Then we binary search the suffix array using the suffix array interval stored at the child, taking additional $O(m+\lg s)$ time. Summing up, the whole search process takes $O(m+\lg\lg\sigma)$ time.


\subsection{Predecessor Queries}
As an additional bonus, the above data structure can be easily augmented to identify the lexicographic predecessor of $P$ if it does not occur in the underlying text $T$. For this we need to construct additional predecessor data structures (again using Prop.~\ref{thm:predecessor}) storing \emph{all} edges outgoing from a node. We also augment each internal node with a pointer to the lexicographically largest leaf in its subtree.

For answering queries, we first try matching $P$ as described above. Now imagine a search for $P[i+1]$ fails at node $v$. Then we query the new predecessor data structure stored at $v$. This identifies a child $w$ of $v$ that is a prefix of the predecessor of $P$. From $w$ we follow the pointer to the rightmost leaf in $w$'s subtree; this is the predecessor of $P$. The time is $O(m+\lg\lg\sigma)$.

\section{New Dynamic Data Structure}

In this section we prove Theorem~\ref{thm:main2}; we call the resulting data structure the \emph{dynamic compacted trie search structure}.

Although we want all bounds to be worst-case, we first start with an amortized version, and then show how to make it worst-case. In both cases, we apply the powerful \emph{exponential search tree} paradigm of Andersson and Thorup~\cite{andersson07dynamic}. In the amortized setting, this results in a fairly simple method. Making it worst-case is a nontrivial challenge, though. Fortunately, instead of using their original deamortization approach, we can follow a significantly simpler method of Bender, Cole, and Raman~\cite{bender02exponential}. First we prove that without loss of generality it is enough to show how to maintain a tree of size $O(\sigma)$. Then we develop a new variant of weighted multiway search trees, which we call \emph{wexponential search trees}. Then we show how to use it to build the new data structure, first in the amortized setting, and then making the bounds worst-case efficient.

The worst-case version of the structure gives us Theorem~\ref{thm:main2}, while the simpler amortized variant achieves the following bounds.

\begin{theorem}
We can maintain a linear-size structure for a text over an alphabet of size $\sigma=n^{O(1)}$ under prepending a new letter in amortized time $O(\frac{\lg^{2}\lg\sigma}{\lg\lg\lg\sigma})$ so that on-line pattern matching concerning the current text can be answered in worst-case time $O(m+\frac{\lg^{2}\lg\sigma}{\lg\lg\lg\sigma})$, for patterns of length $m$, where all complexities are deterministic.
\end{theorem}

\begin{proof}
We maintain the suffix tree. As we aim for amortized time bounds, we can implement the suffix tree oracle using Weiner's algorithm. This requires storing for each node the reverse suffix links in a dynamic predecessor structure implemented using Prop.~\ref{thm:predecessor3} and gives (amortized) $f(n,\sigma)=O(\frac{\lg^{2}\lg\sigma}{\lg\lg\lg\sigma})$. Using such oracle we locate the edge and update it using the amortized version of our structure (Sect.~\ref{sec:amortized search}) in $O(\frac{\lg^2\lg\sigma}{\lg\lg\lg\sigma})$ time.
\end{proof}

\subsection{Reducing the Tree}

In this subsection we show how to reduce the general case of maintaining a tree of size $n$ to maintaining a collection of smaller trees of size $O(\sigma)$. More precisely, we prove the following lemma.

\begin{lemma}
\label{thm:reduction}
Maintaining a tree of size $n$ can be reduced to maintaining a collection of smaller trees with linear total size and each tree being of size $O(\sigma)$ such that each query concerning the original tree can be either answered immediately or translated into at most one query about one of the smaller trees, and each update in the original tree can be translated into a constant number of updates in some of the smaller trees. The former translation will be performed in deterministic worst-case $O(m+\frac{\lg^{2}\lg \sigma}{\lg\lg\lg \sigma})$ time, and the latter in $O(\frac{\lg^{2}\lg \sigma}{\lg\lg\lg \sigma})$.
\end{lemma}

While we aim for worst-case bounds, we start with an amortized version of the above formulation, which is easier to describe, and then explain how to deamortize it. In both cases the idea will be, again, to divide the nodes into heavy and light with parameter $s:=\sigma$, but now we slightly modify the definition so that all nodes with at least $s$ leaves are heavy, all nodes with at most $\frac{s}{2}$ leaves are light, and the remaining nodes can be either light or heavy.

For all light nodes such that the parent is heavy, we maintain a separate smaller tree of size $O(\sigma)$, and explicitly store its number of leaves. For each node $v$ we store a dynamic predecessor structure implemented using Prop.~\ref{thm:predecessor3} storing all of its children. If the node is branching, we keep all its heavy children in an array $A_{v}$ of size $\sigma$, and otherwise we simply keep a pointer to the unique heavy child of $v$, if any. Given a query, we simply use the arrays or pointers to descend as far as it is possible, and then use the dynamic predecessor structure to either locate the small tree where we should continue, or terminate. (So far, the ideas are similar to the static case.) For updates, each light node stores a pointer to the root of the small tree it belongs too. Then after adding a new leaf we follow the pointer to $r$ and increase the counter there. As soon as it becomes equal to $s$, we traverse the whole subtree rooted there and make each node with at least $\frac{s}{2}$ leaves heavy and updating pointers to the root for all nodes that have less of them. Observe that all new heavy nodes are nonbranching, but it might happen that $r$'s parent becomes branching and we need to construct its array. The traversal takes $O(s)$ time, building the array can be done in $O(\sigma)$ time, hence the whole complexity is $O(s)$. We can amortize it by maintaining an invariant that each light node with a heavy parent has $\max(0,\ell-\frac{s}{2})$ credits available, where $\ell$ is the number of leaves in the corresponding subtree. Whenever we insert a new leaf, we need to allocate one new credit. After each traversal we get rid of all light nodes with at least $\frac{s}{2}$ leaves, hence we are free to use all $s$ credits to amortize the expensive operation.

Deamortizing the above method can be done by performing the traversal (and constructing the array, if necessary) incrementally over $\frac{s}{2}$ insertions in the small tree rooted at $r$, starting when its counter reaches $s$. The whole procedure consists of 3 phases. 
\begin{enumerate}

\item We construct the array $A$ for $r$'s parent over $\frac{s}{6}$ insertions, which requires allocating and initializing the memory, and then iterating through all of $r$'s brothers. Note that while we are doing so, new brothers can appear, hence we should keep all children of a node in an unsorted linked list, and append each new child there. The speed of the initialization and iteration should be chosen so that we have enough time to process up to $\sigma$ elements. It might happen that more than one child of a node wants to start the iteration, so each node should have an unique process which constructs the array, and whenever an insertion occurs, we execute a constant number of steps at the parent's process.


\item We run a depth-first search to determine the sizes of all subtrees rooted at the nodes of the small tree. The search is performed incrementally over $\frac{s}{6}$ insertions, again, with the speed of the simulation adjusted so that there is enough time to process a tree of size $\frac{4}{3}s$, as new leaves can appear while we are still in the middle of the traversal, and we might already have $\frac{7}{6}s$ leaves in the subtree at $r$ to start with.

\item We again traverse the tree rooted at $r$ in a depth-first fashion, and make all nodes of size exceeding $\frac{s}{2}$ heavy, where by size we actually mean the size computed in the previous phase. During this traversal we also update the pointers to the root of the small tree a given node belongs to, and rebuild all those smaller trees. More precisely, at node $v$ we set its pointer and insert $v$ into the corresponding smaller tree in the same step. Again, the speed of the simulation is adjusted so that we can process a tree of size $\frac{3}{2}s$ over $\frac{s}{6}$ insertions. Note that as the rebuilding is done in an incremental manner, it might happen that we will update the old smaller tree while the simulation is still running. Nevertheless, in such a case we will sooner or later encounter the place where the update occurred, and change the corresponding new smaller tree accordingly. Notice that all new small trees are of size less than $s$ till the very end of the simulation.
\end{enumerate}

Each query requires probing at most $O(m)$ static dictionaries, one dynamic predecessor structure, and at most one small tree. Each update, on the other hand, requires inserting a new key to one dynamic predecessor structure, and updating a constant number of smaller trees, hence the claimed time bounds follow. The linear space bound follows from the observation that each of the original nodes occurs in at most two smaller trees.

\subsection{Weighted Exponential Trees}

In this section we prove the following theorem, which might be of independent interest.

\begin{theorem}
\label{thm:wexponential}
There is a linear-size data structure that allows storing a collection of weighted sorted keys from an ordered universe of size $u$ so that predecessor search takes $O(\lg\frac{\lg W}{\lg w}\frac{\lg\lg u}{\lg\lg\lg u})$ time, where $W$ is the current total weight of all elements, and $w$ is the weight of the predecessor found. Inserting a new element of weight $1$ takes $O(\lg\lg W)$ time, and increasing by one the weight of an element of weight $w$ takes $O(\lg\frac{\lg W}{\lg w})$ time. All bounds are deterministic worst-case.
\end{theorem}

Whenever we insert a new element, we get its handle in return, which then can be used to increase the weight by one. If we do not have this handle, increasing the weight requires performing a search first.

The proof of the theorem is based on the beautiful idea of Bender, Cole, and Raman~\cite{bender02exponential}, who have shown how to significantly simplify the deamortization presented by Andersson and Thorup~\cite{andersson07dynamic}. We start with the amortized version of the theorem, and later show how to make it worst-case efficient.

Let $f(\ell)=\lfloor 2^{(\frac{3}{2})^{\ell}}\rfloor$, so $\ell=\Theta(\lg\lg f(\ell))$. We define a weighted multiway search tree with the degrees increasing doubly-exponentially along any leaf-to-root path, which we call a wexponential search tree of level $\ell$, in the following recursive manner:
\begin{enumerate}
\item the (explicitly stored) current total weight $W$ of all elements is less than $2f(\ell+1)$, and if $W\geq 2f(\ell)$ the tree is \emph{proper},\label{enum:total}
\item we store a static predecessor search structure (implemented using Prop.~\ref{thm:predecessor2}) containing a subset $S=\{e_{1},\ldots,e_{|S|}\}$ of all elements, which we call the \emph{splitters},
\item the remaining elements are split into $X_{0},\ldots,X_{|S|}$ such that $e_{i}$ is between $X_{i-1}$ and $X_{i}$,
\item the total weight of all elements in $\{e_{i}\}\cup X_{i}\cup\{e_{i+1}\}$ exceeds  $f(\ell)-f(\ell-1)$ for all $i=1,2\ldots,|S|-1$,
\item each $X_{i}$ is stored in a wexponential search tree of level $\ell-1$, which are called the \emph{children},\label{enum:children}
\item for each $i$ the predecessor structure stores a bidirectional link to the child storing $X_{i}$, and additionally a link to the leftmost child storing $X_{0}$ is kept.
\end{enumerate}
Observe that if a weight of an element is at least $2f(\ell)$,
it must belong to $S$. Note also that the definition requires that some of the $X_i$'s may be empty.
Furthermore, we can bound the size of $S$ as follows:
\begin{equation}\label{eqn:splitters}
|S| \leq 2\frac{2f(\ell+1)}{f(\ell)-f(\ell-1)} = O(2^{(\frac{3}{2})^{\ell+1}-(\frac{3}{2})^{\ell}})=O(2^{\frac{1}{2}(\frac{3}{2})^{\ell}})=O(f^{\frac{1}{2}}(\ell)) .
\end{equation}
Updating the structure will be done in a bottom-up order. In other words, we will assume that the children are valid wexponential search trees of level at most $\ell-1$, and show how to ensure that their parent is a valid tree  of level $\ell$.  Inserting a new element of weight one or increasing the weight of some element by one might cause the total weight to become $2f(\ell+1)$. As soon as we detect such a situation, we \emph{split} the tree into two by choosing an element $e_{i}$ from $S$. To choose this element we look at the sets of its predecessors $P_{i}=X_{0}\cup\{e_{1}\}\cup\ldots\cup\{e_{i-1}\}\cup X_{i-1}$ and successors $S_{i}=X_{i}\cup\{e_{i+1}\}\cup\ldots\cup\{e_{|S|}\}\cup X_{|S|}$. As the total weight is $2f(\ell+1)$, and the weight of any $X_{i}$ is less than $2f(\ell)$ (by conditions \ref{enum:children} and \ref{enum:total}), we can always select an $i$ so that the weight of both $P_{i}$ and $S_{i}$ is less than $f(\ell+1)+f(\ell)$,
but the weight of both $P_{i}\cup\{e_{i}\}$ and $\{e_{i}\}\cup S_{i}$ is at least $f(\ell+1)-f(\ell)$; see Fig.~\ref{fig:split}. We construct two new wexponential search trees of level $\ell$ containing all elements in $P_{i}$ and $S_{i}$, which is possible as their total weights are at most $f(\ell+1)+f(\ell)<2f(\ell+1)$. This requires constructing static predecessor search structures containing $e_{1},e_{2},\ldots,e_{i-1}$ and $e_{i+1},e_{i+2},\ldots,e_{|S|}$, respectively, and making each wexponential search tree of level $\ell-1$ a child of the former or the latter new tree. Notice that we don't have to rebuild the smaller trees, as simply redirecting the pointers to already existing structures is enough. Then we look at the parent of the structure that we are splitting. If there is none, we simply create a new proper wexponential search tree of level $\ell+1$ with just one splitter $e_{i}$. Otherwise we add $e_{i}$ to its set of splitters, and store pointers to the two newly created trees of level $\ell$ in its predecessor structure, which needs to be rebuilt, or \emph{refreshed}. The whole splitting and refreshing process is a very local procedure, as it requires rebuilding the static predecessors structures only for the tree and its parent, while the descendants and further ancestors are kept intact.

\begin{figure}
\centering
\includegraphics[scale=1]{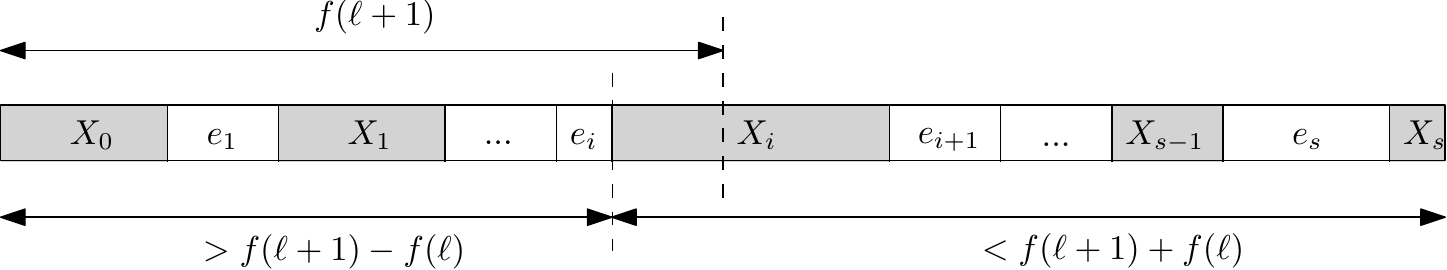}
\caption{Finding a suitable new splitter $e_{i}$. The smaller of the two parts has size at least $f(\ell+1)-f(\ell)$, and hence the larger one at most $f(\ell+1)+f(\ell)$.}
\label{fig:split}
\end{figure}

Now we can describe how to query and update the wexponential search tree.
\begin{description}
\item[predecessor search:] First we use the static predecessor search structure, which takes $O(\frac{\lg\lg u}{\lg\lg\lg u})$. If the query element belongs to $S$, we are done. Otherwise we recurse in the smaller structure.

\item[insert:] Using the static predecessor structure we locate the smaller structure the new element belongs to, and insert it there recursively. Then we increase $W$ by one and split the tree, if necessary. For each new element we allocate a record storing a link to the tree where we have used it as a splitter, and return a pointer to this record. Whenever some static predecessor search structure is rebuilt, or an element is moved up to the parent, we update the record.

\item[increasing:] We locate the tree where the element is a splitter using the record in constant time. Then we increase the total weight by one and split the tree, if necessary, and move up to its parent.
\end{description}

To analyze the complexity of this implementation, we first need a simple lemma.

\begin{lemma}
Consider a proper wexponential search tree of total weight $W$ and an element of weight $w$. The element is used as a splitter at depth $O(\lg\frac{\lg W}{\lg w})$.
\end{lemma}

\begin{proof}
As the tree is proper, its level is $\Theta(\lg\lg W)$. On the other hand, if an element of weight $w$ belongs to a subtree of level $\ell$, but is not chosen as a splitter there, then $w< 2f(\ell)$,
which is equivalent to $\ell = \Omega(\lg\lg w)$. Hence the maximum possible difference between the level of the whole tree and the level of the subtree where the element is used as a splitter is $O(\lg\frac{\lg W}{\lg w})$.
\end{proof}

The above lemma shows that worst-case complexity of predecessor search in a wexponential search tree is $O(\lg\frac{\lg W}{\lg w}\frac{\lg\lg u}{\lg\lg\lg u})$, where $W$ is the total current weight and $w$ is the weight of the predecessor found. Indeed, consider this (unique) predecessor: it must be stored at depth $O(\lg\frac{\lg W}{\lg w})$, and traversing each level requires one query to the static predecessor structure. The complexity of both insert and increase is more tricky to estimate, as we might need to repeat the expensive splitting procedure a couple of times. Nevertheless, insert traverses all $O(\lg\lg W)$ levels, and increase traverses just $O(\lg\frac{\lg W}{\lg w})$ levels. At each of those levels we might need to split the tree, which takes a lot of time, but cannot happen very often. We start with an amortized bound.

For each wexponential tree we maintain an invariant that we have at least $\max(0,W-f(\ell+1))$ credits allocated there, where $W$ is the current total weight and $\ell$ is the level. As long as we don't split, the invariant is easy to maintain: whenever we move from a tree to its parent during an insert or increase and add one to its total weight, we put an additional credit there. Now consider splitting a tree of total weight $W=2f(\ell+1)$. We need to rebuild the static predecessor structures at both the tree (or, more precisely, at the two new trees) and its parent, hence we need to apply Prop.~\ref{thm:predecessor2} to a set of size which we bounded in (\ref{eqn:splitters}) by $f^{\frac{1}{2}}(\ell+1)$, which takes $O(f(\ell+1))$ time. On the other hand, we have $W-f(\ell+1)=f(\ell+1)$ credits available, and for each of the two new trees we need to keep just $\max(0,f(\ell+1)+f(\ell)-f(\ell+1))=f(\ell)$ of them (recall that the larger of the new trees has weight at most $f(\ell+1)+f(\ell)$). Hence we can use the remaining $f(\ell+1)-2f(\ell)=\Theta(f(\ell+1))$ credits to pay for the reconstruction.

As usual, deamortizing the running time requires more care. Fortunately, we can fairly closely follow the method of Bender, Cole, and Raman~\cite{bender02exponential}.  The only part of the implementation that is not worst-case efficient is splitting. Hence instead of immediately splitting a tree as soon as its weight becomes $2f(\ell+1)$, we perform it incrementally over $\frac{1}{2}f(\ell)$ updates concerning the tree or one of its descendants, starting when the weight becomes $2f(\ell+1)$.
Furthermore, instead of refreshing the parent as soon as we have two new trees, we use the bidirectional pointer to replace the link to the old tree kept there by a record containing the element $e_{i}$ used for partitioning and the links to the new two trees. As long as the parent is not fully refreshed and the new trees are still linked from the same record, we call them \emph{twins}.
Similarly, refreshing is performed incrementally over $\frac{1}{2}f(\ell)$ updates concerning the tree or one of its descendants, starting whenever we notice that the weight is a multiple of $\frac{1}{2}f(\ell)$ not exceeding $2f(\ell+1)-\frac{1}{2}f(\ell)$. This ensures that we never need to split a twin, as between splitting a child and splitting one of the two new children created as a result, the tree will be fully refreshed, hence we avoid a situation where we already keep a record containing two links, and now we would need to replace one of them by such a record again. Observe that we never try to refresh and split a tree at the same time, hence there is no interference between those two operations.

Splitting requires first choosing a good splitter $e_{i}$, which can be done in a single left-to-right sweep through the contents of the static predecessor search structure, and then building new static predecessor search structures containing $\{e_{1},e_{2},\ldots,e_{i-1}\}$ and $\{e_{i+1},\ldots,e_{|S|}\}$. This must be done in the background, hence it might happen that some updates concerning the already seen part occur. Hence we can only guarantee that the total weight of $P_{i}$ and $S_{i}$ is at most $f(\ell+1)+\frac{3}{2}f(\ell)$, as there might be up to $\frac{1}{2}f(\ell)$ of such updates. An additional complication is that the weight of $P_{i}$ or $S_{i}$ might be so that we would expect the corresponding new tree to be undergoing refreshing at the moment, which is the reason we have chosen to refresh when the weight is a multiple of $\frac{1}{2}f(\ell)$ instead of simply taking a multiple of $f(\ell)$. We never skip two consecutive moments when we should start refreshing, as the next split starts not sooner than after $f(\ell+1)-\frac{3}{2}f(\ell)\geq f(\ell)$ updates, hence it still never happens that we try to split a twin.

Rebuilding the static predecessor search structures in the background requires storing two versions, old and new, at the same time. More precisely, we have the old version, which we use for navigation and answering any query, and the new one that is being built. The corresponding elements in both versions are linked to each other, and any update is performed in both of them. Then, when the new version is ready, we simply discard the old one in constant time. Discarding can be done by storing timestamps that can be used to determine which links are still valid, and which should be actually null at the moment. The definition of a wexponential tree must be slightly relaxed by saying that its current total weight never exceeds $2f(\ell+1)+\frac{1}{2}f(\ell)$, so the bounds on the complexity of search, and both insert and increase (excluding the cost of splitting) still hold, and the theorem follows.

\subsection{Amortized Version}
\label{sec:amortized search}

In this section we develop an amortized version of our dynamic compacted trie search structure.

For each node $v$ we keep two separate structures. The first is a static dictionary implemented using Prop.~\ref{thm:hashing} that stores edges leading to all children $v_{i}$ of size ``similar'' to the size of $v$. The second is a dynamic predecessor data structure implemented using Theorem~\ref{thm:wexponential}. To make the notion of ``similar size'' more precise, define the \emph{weight} of a node to be the number of leaves in its subtree, and its \emph{level} to be $\ell$ when the weight belongs to $[f(\ell),2f(\ell+1))$ ($[f(\ell),3f(\ell+1))$ in the worst-case version presented in the next section),
where $f(\ell)=\left\lfloor 2^{(\frac{3}{2})^{\ell}} \right\rfloor$, again. Clearly, the weights and hence the levels along any root-to-leaf path are nonincreasing. We define the \emph{fragment} of node $v$ to be the maximal subtree containing $v$ and consisting of nodes of the same level. The root $r$ of a fragment containing $v$ is its lowest ancestor of level $\ell$, but with parent of level at least $\ell+1$. For each such fragment we store the root $r$ and a list of bidirectional links to all its nodes, and a counter with the number of leaves in the subtree rooted at $r$. For each node $v$ we store all edges leading to children of $v$ of the same level $\ell$ in a static dictionary. All edges leading to children of smaller levels are kept in a wexponential search tree. We would like the weights in this structure to be the same as the weights of the corresponding nodes, but for technical reasons we maintain a weaker condition, namely a node of weight $w$ has weight from $[\sqrt{w},w]$ in the wexponential search tree stored at its parent. First we show that this relaxation is enough to guarantee good bounds on the search time.

\begin{lemma}
\label{thm:search}
Traversing any path of length $m$ starting at the root takes just $O(m+\frac{\lg^{2}\lg\sigma}{\lg\lg\lg\sigma})$ time.
\end{lemma}

\begin{proof}
At each node we first use the static dictionary to check if the next edge we would like to traverse is stored there. If so, we continue. Otherwise we query the wexponential search tree. In order to bound the complexity of the whole procedure, we only have to bound the total time of the latter steps, as the former sum to at most $O(m)$. Observe that whenever we query the wexponential search tree, we either terminate, or decrease the current level, which is already enough to get a bound of $O(\frac{\lg^{3}\lg\sigma}{\lg\lg\lg\sigma})$. To get the claimed complexity, let $W_{1},W_{2},\ldots,W_{k}$ be the weights of nodes where we query a wexponential search tree. Similarly, let $w_{1},w_{2},\ldots,w_{k}$ be the corresponding weights of nodes that we find there (note that $W_{i}$ is not necessarily stored in our implementation, but $w_{i}$ certainly is). Let $w'_{i}$ be the weight of the elements of the wexponential search tree corresponding to $w_{i}$, and $W'_{i}$ the total weight of this structure. Then we have the following inequalities:
\begin{eqnarray}
w'_{i}  &\in& [\sqrt{w_{i}},w_{i}] \label{eq:wprime}\\
W'_{i} &\leq& W_{i} \label{eq:Wprime}\\
W_{i+1} &\leq& w_{i}\label{eq:W}
\end{eqnarray}
and the time of this part is order of
$$
\frac{\lg\lg\sigma}{\lg\lg\lg\sigma}\sum_{i} \lg\frac{\lg W'_{i}}{\lg w'_{i}}
\stackrel{\eqref{eq:wprime},\eqref{eq:Wprime}}{\le} \frac{\lg\lg\sigma}{\lg\lg\lg\sigma}(\lg\lg\sigma+\sum_{i}\lg\frac{\lg{W_{i}}}{\lg w_{i}})
\stackrel{\eqref{eq:W}}{=} \frac{\lg\lg\sigma}{\lg\lg\lg\sigma}(\lg\lg\sigma+\lg\frac{\lg{W_1}}{\lg w_k})
\le \frac{\lg^{2}\lg\sigma}{\lg\lg\lg\sigma}\ .
$$
\end{proof}

Whenever we add a new leaf, the weights of all nodes on its path to the root increase by one. We iterate over all fragments above the new leaf and increase their counters. Iterating is done in $O(\lg\lg\sigma)$ time by starting at the leaf and repeatedly jumping to the root of the current fragment by following the bidirectional link stored at each node. To maintain the invariant that the weights on any path are nonincreasing, we actually first construct a list of all fragments, and then update their counters one-by-one in a top-down order. For each root that we consider we need to update its corresponding weight in the wexponential search tree at its parent, which takes at most $O(\lg\frac{\lg W}{\lg \sqrt{w}})=O(1+\lg\frac{\lg W}{\lg w})$ time, where $w$ is the weight of this root, and $W$ is at most the weight of its parent. Summing up over all roots, as in the proof of Lemma~\ref{thm:search}, we get a telescoping expression which is at most $O(\lg\lg\sigma)$.

During this procedure it might happen that we increase the weight of some root $r$ to $2f(\ell+1)$, and hence need to increase its level. Maintaining the invariant in such situation is a very costly procedure, and we need to somehow amortize this cost. We start at $v:=r$ and descend down to its (unique) child of weight exceeding $f(\ell+1)$ as long as possible.

Note that we don't actually store the weight of each node, but given a weight of $v$ with exactly one child of level $\ell$, we can compute the weight of this child
by iterating through all other children, which are roots of their fragments, and hence have up-to-date counters available. Note that if there is more than one child of level $\ell$, there cannot be any child of weight exceeding $f(\ell+1)$. We call the traversed path the \emph{tail}, and increase the level of all its nodes, see Fig.~\ref{fig:tail}. Then maintaining the invariants requires four steps.

\begin{figure}
\centering
\includegraphics[scale=1]{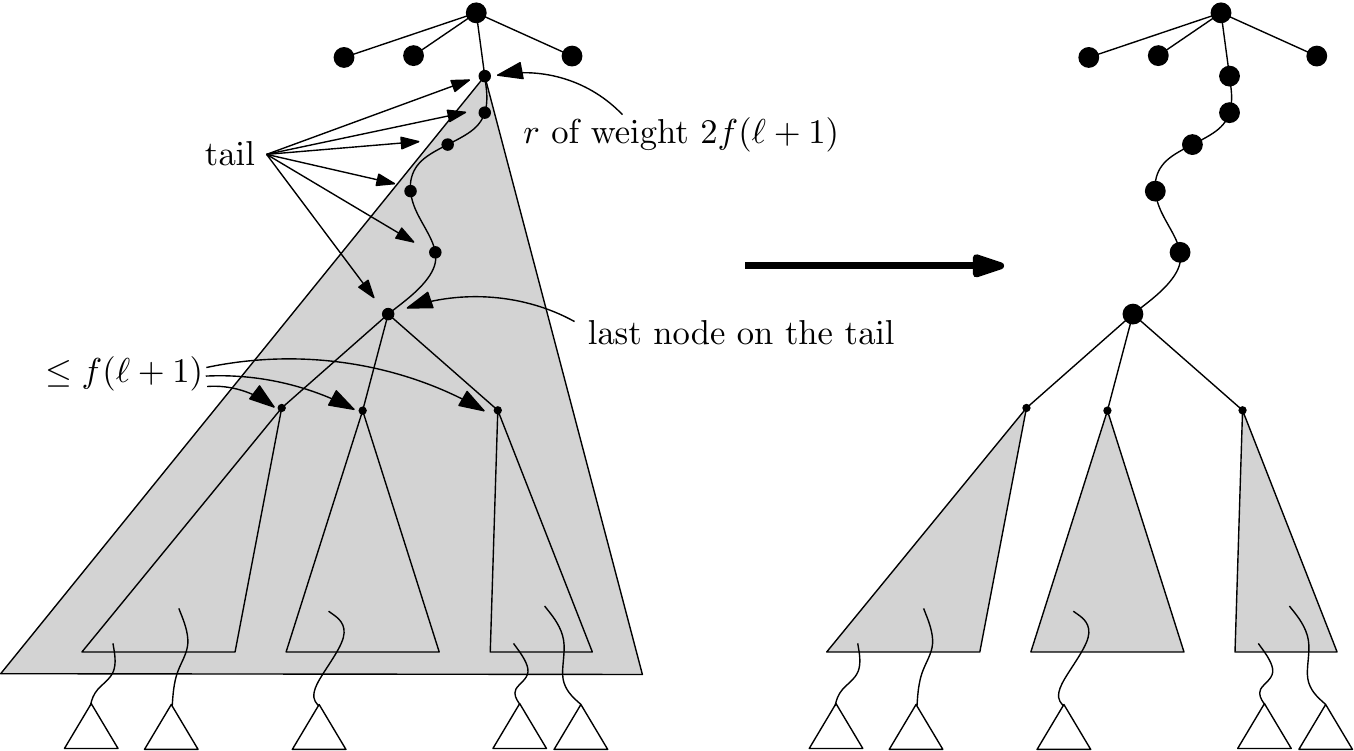}
\caption{Bumping the levels of tail nodes.}
\label{fig:tail}
\end{figure}

\begin{enumerate}
\item If the parent of $r$ is of level $\ell+1$, we must rebuild its static dictionary in order to include a new element there.
\item We move all nodes on the tail from the list of the current fragment to either the list of the fragment corresponding to the parent of $r$, if its level is $\ell+1$, or a new fragment.
\item If the last node on the tail has more than one child of level $\ell$, bumping its level to $\ell+1$ splits the current fragment into more than one. We need to iterate through all nodes there, and partition them accordingly creating new fragments. Note that creating a new fragment requires computing the weights of their roots, which can be done by iterating through children of smaller levels for all nodes in the current fragment.
\item For each child of level $\ell$ of the last node on the tail we must add a new element to the wexponential search tree. Note that at this point all those nodes are roots of their fragments, hence have their weights computed.
\end{enumerate}

Notice that all those steps are local in the sense that they modify only the nodes in the current fragment. To bound the time taken by the whole procedure,
we allocate credits to fragments, making sure that a fragment of weight $w$ and level $\ell$ has $\max(0,w-f(\ell+1))$ credits available.
Whenever we split an edge and create a new leaf, we allocate one credit for each of the at most $O(\lg\lg\sigma)$ fragments above. Then when we are increasing the level of $r$, its weight is $2f(\ell+1)$, so we have $f(\ell+1)$ credits available, and because of the way we defined the tail, we can spend all of them, as all new fragments of level $\ell$ will be of weight at most $f(\ell+1)$ after the update. We can bound the time required for maintaining the invariants, which we call \emph{promoting} at $r$, as follows.

\begin{enumerate}
\item Rebuilding the static dictionary takes $O(\frac{f(\ell+2)}{f(\ell+1)}\log^2\log\frac{f(\ell+2}{f(\ell+1)})=O(\frac{f^2(\ell+2)}{f^2(\ell+1)})$. We call this refreshing the parent of $r$.
\item There are at most $4f(\ell+1)-1$ nodes in the subtree of $r$, hence traversing the tail, including the time taken to compute the weight of all nodes there, takes $O(f(\ell+1))$.
\item There are at most $2f(\ell+1)$ nodes in the current fragment, hence the nodes in the current fragment can be partitioned into new fragments in $O(f(\ell+1))$ time. This includes the time necessary to compute the weights of their roots, as in the worst-case we iterate through at most $4f(\ell+1)-1$ nodes in the subtree of $r$.
\item We must insert at most $\frac{2f(\ell+1)}{f(\ell)}$ elements into the wexponential tree stored at the last node of the tail, and inserting an element of weight $w$ is done by first adding a new element of weight one, and then increasing its weight repeatedly $\sqrt{w}$ times.
\end{enumerate}

Hence the total number of credits required by the promoting is:

$$ O\left(\frac{f^{2}(\ell+2)}{f^{2}(\ell+1)} +  f(\ell+1) + \frac{f(\ell+1)}{f(\ell)} \sqrt{f(\ell)}\lg\lg f(\ell+1)\right) = O\left(\frac{f^{2}(\ell+2)}{f^{2}(\ell+1)} + f(\ell+1)\right). $$

Hence we only need to ensure that $\frac{f^{2}(\ell+2)}{f^{2}(\ell+1)} \leq f(\ell+1)$, which is equivalent to $f^{2}(\ell+2) \leq f^{3}(\ell+1)$, and then $2(\frac{3}{2})^{\ell+2} \leq 3(\frac{3}{2})^{\ell+1}$, hence by the choice of $f$ we always have enough credits to amortize the update.

\subsection{Worst-Case Version}

In this section we develop the final worst-case efficient version of our dynamic compacted trie search structure.

The only non worst-case efficient part of the previous implementation is increasing the level of a root $r$. Instead of traversing the tail and updating all the structures as soon as its level reaches $2f(\ell+1)$, we will execute those operations incrementally over the next $f(\ell+1)$ insertions in the subtree rooted at $r$, and relax the condition on the weight of a node of level $\ell$ by saying that it shouldn't exceed $3f(\ell+1)$. As selecting the tail is done incrementally, we redefine it to be the maximal sequence of nodes of weight at least $s$ at the moment we started the process, which requires storing at each edge a timestamp denoting its creation time.

First of all, we must make sure that there is at most one promoting process per fragment. Even though we have already chosen the tail, future insertions might increase the weight of some additional nodes to more than $f(\ell+1)$. Nevertheless, as we execute the procedure over just $f(\ell+1)$ insertions, no node in the current fragment (or in one of the new fragments) which doesn't belong to the tail can reach the weight of $2f(\ell+1)$ before we are done. Furthermore, there is some interaction between different fragments. While we are still promoting at $r$, new nodes might be added to the corresponding fragment. Also, different children of a node might need to be refreshing it at overlapping periods of time.

We choose the speed of the simulation so that there is enough time to process a fragment consisting of $3f(\ell+1)$ nodes over $f(\ell+1)$ insertions. When splitting the current fragment into more than one, we run a depth-first search to determine the new fragments. As soon as we reach a node, we set the link to its (new) fragment. Then when a new node appears, its parent is either already processed and hence has the correct link set, or will be seen later, and we will notice the new child then. The same reasoning works for inserting the new elements into the wexponential tree stored at the last node of the tail. Refreshing a node is  a quite different issue, though, as we must somehow deal with the problem that many children might need to refresh the same parent. Each node of level $\ell$ stores a list of all its children of weight at least $f(\ell)$, where we simply append a child as soon as its weight becomes large enough. As a part of our simulation we run the refreshing process. In other words, each child of a node can be potentially running a deferred process refreshing its parent. The first step of such process is making a read-only copy of the current list. As this list only grows, instead of making a physical copy we can simply store the current last element, which can be done in constant time. Then we build a new static dictionary containing all elements on this read-only copy over the insertions in the subtree rooted at the child.  As soon as the construction finishes, we substitute the old dictionary in constant time by replacing one pointer. There is just additional detail: we should first check if we are really replacing an older version, i.e., we can simply look at the number of elements stored there. This is because it might have happened that there is a refreshing process which starts later, yet finishes earlier (as there are more insertions to the corresponding subtree). This ensures that when the weight of $r$ reaches $3f(\ell+1)$, the static dictionary stored a its parent surely contains its edge (and, potentially, also some more recent edges).

\section{Conclusions and Open Problems}
\label{sect:conclusions}
We gave data structures for improved deterministic worst-case searching and update times in suffix trees and compacted tries in general. Whereas previous data structures either ignored the alphabet size completely \cite{andersson07dynamic} or had at least a logarithmic dependency on the alphabet size, our new structures have only a doubly-logarithmic dependency on the alphabet size. The obvious open questions are if better running times are possible, or if non-trivial lower bounds exist?

\section*{Appendix}
\appendix
\section{Better suffix tree oracle}

In this section we show how to implement a suffix tree oracle with $f(n,\sigma)=O(\lg\lg n + \frac{\lg^2\lg \sigma}{\lg\lg\lg \sigma})$. At a high level, we follow the approach of Breslauer and Italiano~\cite{Breslauer11fringe}, which in turn is based on Weiner's algorithm~\cite{weiner73linear}. The main difficulty is that both methods are only efficient when the alphabet is of constant size, which is not our case, and some additional ideas are required. Nevertheless, we are able to achieve $f(n,\sigma)=O(\lg\lg n+\lg\sigma)$ by rather simple means. Then we decrease this bound even further by applying the structure of Kopelowitz~\cite{Kopelot12indexing}, which needs some minor tweaking as we are interested in a deterministic solution.

For each node $u$ of the suffix tree we define its $a$-link $W_a(u)$ if the suffix locus $v=au$ exists. If $v=au$ is a node, then the link is called a hard $a$-link and points to this node, and otherwise a soft $a$-link and points to the lower end of the corresponding edge. Hard links correspond to inverse suffix links. Assuming that we maintain the suffix tree using Weiner's algorithm, the following properties hold.

\begin{observation}
\begin{enumerate}[(a)]
\item\label{obs:monotone} If $W_a(u)$ is defined and $v$ is an ancestor of $u$, then $W_a(v)$ is defined, too. In other words, the sets of links are monotone when we traverse the tree starting from the root.
\item\label{obs:change} As soon as $W_a(u)$ becomes hard, it stays hard and cannot change. On the other hand, a soft link may change or become hard at some point.
\item\label{obs:reconstruct} $W_a(u)$ is defined if and only if there is a descendant $v$ of $u$ with the hard $a$-link defined, and in fact a stronger statement holds: $W_a(u)$ is defined iff $W_a(u)=W_a(v)$, where $W_a(v)$ is a hard $a$-link, and $v$ is the (unique) highest descendant of $u$ with the hard $a$-link defined. (Uniqueness follows because if two nodes have a hard $a$-link defined, then their lowest common ancestor also has it.)
\end{enumerate}
\end{observation}

Breslauer and Italiano show that to implement a suffix tree oracle, one needs to store the suffix tree such that the following four operations can be performed:

\begin{enumerate}
\item given a letter $a$ and a node $v$, locate the lowest ancestor $u$ of $v$ with the $a$-link defined, and retrieve its $W_a(u)$,
\item make a soft link hard,
\item create or change a soft $a$-link,
\item split an edge $(u,v)$ of the suffix tree, and for each existing $a$-link $W_a(v)$ (both soft and hard) create a new soft $a$-link $W_a(v'):=W_a(v)$, where $v'$ is the new middle node and $v$ is the lower end of the split edge.
\end{enumerate}

After prepending a single letter to the text a constant number of operations of each type is executed.
The most problematic is the last one, as
there might be many defined $a$-links from $v$. 
To overcome this, we will maintain the soft links just at some carefully chosen nodes. For the remaining nodes we will exploit (\ref{obs:reconstruct}) to recover the missing soft links whenever we need to access them. Furthermore, during the execution of the algorithm it might happen that we haven't yet updated some soft links. Nevertheless, we always update them in a top-bottom order, so (\ref{obs:monotone}) still holds, and all hard links are eagerly created as soon as they should appear. The details are as follows.

\subsection{Tree Decomposition}
We (again) divide the nodes into heavy and light, but this time we define the \emph{weight} of a node to be the number of leaves in its subtree plus the number of hard links defined there (which due to (\ref{obs:change}) can only increase). We choose parameter $s:=\sigma$ and say that all nodes of weight at least $3s$ are heavy, all nodes of weight at most $s$ leaves are light, and the remaining nodes can be either light or heavy.

Heavy nodes store \emph{all} Wlinks in a dynamic dictionary (implemented using Prop.~\ref{thm:predecessor3}), whereas light nodes store only the hard ones. 
Additionally, we maintain a structure which allows us to locate for a given heavy node $v$ its lowest ancestor $u$ with the $a$-link defined (point (1) above). The details on how to update this structure will be given later; we first describe how it can be used. We will consider the two possible cases where the answer is a heavy node or when it is light.

\subsection{The Answer is a Heavy Node}
First we need a result on a special case of the marked ancestor problem.

\begin{lemma}[Fringe marked ancestor (FMA) structure~\cite{Breslauer11fringe}]
  \label{lem:fringe}
  We can maintain a tree $T$ on $n$ nodes which are either marked or unmarked under the following operations:
  \begin{enumerate}
  \item inserting a new node in the middle of an edge, which adopts the marked status of its parent,
  \item inserting a new leaf, which is initially unmarked,
  \item marking the root or a node whose parent is already marked,
  \item locating the lowest marked ancestor of a given node.
  \end{enumerate}
\end{lemma}

Due to the way marking is done the marked nodes form a connected subtree containing the root. Hence, the answer to queries on unmarked nodes will always lie on the \emph{fringe} of this subtree. All operations cost $O(\lg \lg n)$ time~\cite{Breslauer11fringe},
which is faster than in the general setting where a lower bound of $\Omega(\frac{\lg n}{\lg\lg n})$ for at least one of the operations exists~\cite{alstrup98marked}

We store a single FMA structure over the suffix tree, where all \emph{heavy} nodes are marked. This allows us to locate for a given node $x$ its lowest heavy ancestor $v$ in $O(\lg\lg n)$ time. Because the heavy subtree can be quite large, we now need to do the following. Define the \emph{induced heavy subtree} as the tree that contains only the heavy leaves and branching heavy nodes. We store this tree explicitly, and couple all edges of the \emph{original} suffix tree with their corresponding edges in the \emph{induced} tree by bidirectional links that allow us to move back and forth between the two trees.

In this induced tree, for each letter $a$ we keep a separate FMA structure, where a node is marked when its $a$-link is defined.
Due to (\ref{obs:monotone}) the marked nodes form indeed a connected subtree containing the root, so Lemma \ref{lem:fringe} can be applied.
Additionally, for each edge of the induced tree we keep an array $\mathrm{lowest}[a]$ storing a pointer to the lowest heavy node (of the \emph{original suffix tree}) that is coupled with the edge and that has the $a$-link defined. Due to (\ref{obs:monotone}) this array gives us the only candidate on the edge of the induced tree.

All FMA structures, arrays, and pointers to edges are enough to locate for a given heavy node $v$ its lowest ancestor $u$ with the $a$-link defined. First we check if $u$ and $v$ belong to the same edge of the induced tree. If not, we replace $v$ by its lowest branching ancestor $v'$ by going to the top node on the edge of the induced tree. Then we locate its lowest branching ancestor $u'$ with the $a$-link defined using the FMA structure. Then if we look at the edge outgoing from $u'$ and leading to the subtree $v$ belongs to, $u$ surely is there.
The edge can be located by retrieving the corresponding letter, and return $\mathrm{lowest}[a]$. The total time taken by those operations is $O(\lg\lg n + \frac{\lg^2\lg\sigma}{\lg\lg\lg\sigma})$, where the second addend is due to the edge retrieval. Furthermore, the total size of the structure is linear, as we need $O(\sigma)$ space per edge, and we keep $\sigma$ FMA structures, each on $O(\frac{n}{s})$ nodes.

\subsection{The Answer is a Light Node}
For each maximal subtree consisting of light nodes we store a separate structure allowing us to quickly locate for a given node $v$ in this subtree its lowest light ancestor $u$ with the $a$-link defined, and retrieve its $W_a(u)$. By (\ref{obs:reconstruct}), we only have to maintain all hard links as long as given $v$ we can quickly locate its highest descendant $u$ with the hard $a$-link defined. Assume that we can maintain a preorder traversal of all nodes in the subtree. Then there is just one candidate for $u$:
the first node after $v$ in the traversal with the hard $a$-link defined. Hence if we maintain for each $a$ a separate predecessor structure where we put all nodes in the subtree with the hard $a$-links defined, we just have to perform one predecessor query. In the simplest form, the structure could be any balanced search tree, giving us a $O(\lg\sigma)$ time bound if we observe that comparing any two elements there reduces to computing the lowest common ancestor, which can be done in worst-case constant case even on dynamic trees~\cite{cole05dynamic}.

We aim to achieve better bounds, though. For this we would like to implement the predecessor structure using Prop.~\ref{thm:predecessor3}, but in order to do so we need to assign integer labels consistent with the preorder traversal to the nodes. This is exactly the POLP problem considered by Kopelowitz~\cite{Kopelot12indexing}. 

\begin{lemma}[Predecessor search on dynamic subsets of an ordered dynamic list~\cite{Kopelot12indexing}]
There exists a data structure for a dynamic ordered list of size $n$ and disjoint subsets such that:
\begin{enumerate}
\item order queries are answered in $O(1)$ worst-case time,
\item creating a new node after a given node takes $O(1)$ worst-case expected time,
\item adding an element into any subset takes $O(\lg\lg n)$ worst-case expected time,
\item predecessor search in any subset takes $O(\lg\lg n)$ worst-case expected time.
\end{enumerate}
\end{lemma}

We apply his method with the following modifications.

\begin{enumerate}
\item We aim to achieve deterministic time bounds, and hence need to replace $y$-fast tries by the structure from Prop.~\ref{thm:predecessor3}. In order to do so, we must verify that the following is possible: given a structure storing a set $S$, we can incrementally perform insertion into $S$ over $O(\lg\sigma)$ steps so that we can perform predecessor queries over the original $S$ while the insertion is not over yet. But this is easy if the structure is an exponential search tree, which is just a multiway search tree where the insertion is done in a bottom-up order. Whenever we replace a node $v$ by $v'$, we store the pointer to the old $v$ at $v'$ together with a timestamp indicating when we should start using the new $v'$ to answer predecessor queries. As there is just one insertion process per structure, we never need to store more than one previous version per node.
\item As the data structure of Prop.~\ref{thm:predecessor3} is of linear size, we do not need the bucketing as in the $y$-fast tries, which simplifies the structure a little bit as there is no interference between splitting of buckets and chunks.
\item All sets stored in the structures should be disjoint to guarantee that whenever a label of some element changes, we need to update just a single element in one of them. For this we define our list to contain all pairs $(u,a)$ such that $W_a(u)$ is a hard link in the natural lexicographical order, i.e., extending the preorder on the nodes.
\end{enumerate}

As a result, we get both queries and updates in deterministic $O(\frac{\lg^2\lg\sigma}{\lg\lg\lg\sigma})$ worst-case time, where by update we mean either splitting an edge or creating a hard link. This is of course assuming that the weight of each light subtree stays at most $s$. 

\subsection{Updating the Structures}
As soon as the weight of some light subtree reaches $2s$, we start the bumping process which is executed over the next $s$ updates there. As in Section~\ref{sec:amortized search}, we define the tail to be the maximal sequence of light nodes of weight at least $s$ starting at the root of the subtree (again, by weight we actually mean the weight at the moment when we started the process). We construct a list of those nodes, and rebuild the structures for all new smaller light subtrees. This takes $O(s)$ time because we defined the weight to include the number of hard links, and can be performed incrementally over $\frac{s}{2}$ updates keeping the structure for the original subtree as long as the procedure is not over yet so that we can answer queries. Then we must make all nodes on the tail heavy, which adds one edge to the tree induced by the branching heavy nodes, hence can potentially make at most one non-branching node branching.

First of all, we must set pointers to this new edge and construct its array in $O(\sigma)$ time. Then we might need to split an existing edge into two, which requires updating the pointers of all nodes there and constructing the array for the two new parts. Here we need to modify our definition of an edge slightly: we don't want its length to significantly exceed $s$. Hence we consider the tree induced by all branching nodes and some \emph{virtual branching nodes} chosen so that each edge is of length at most $\frac{3}{2}s$ and either connects at least one real branching node, or is of length at least $\frac{s}{2}$. In other words, whenever we have a long edge, we add a new virtual branching node in the middle. Now whenever we split an edge into two parts, we need to iterate over at most $s$ nodes $v_1,v_2,\ldots,v_k$. For each of the two resulting parts we need to construct the arrays and change the pointers to the current edge. This needs to be done incrementally over the remaining $\frac{s}{2}$ updates into the light subtree. The first difficulty is as soon as we modify the pointer to the current edge of some $v_i$, the array of this edge must be ready. So it seems that we should first construct the arrays, and then set the pointers. But then some new soft links could appear, and we wouldn't yet know which array should be updated. Hence we temporarily store pointers to both the old and the new edge at each node, updating both of them whenever a soft link changes or appears, and proceed in phases as follows.

\begin{enumerate}
\item Initialize each of the new arrays to contain just nulls.
\item Insert the new middle branching node into at most $\sigma$ FMA structures.
\item Iterate over $v_k,v_{k-1},\ldots,v_1$ and set the pointer to the new edge for each of them in this order. Furthermore, if $W_a(v_i)$ is defined but $W_a(v_{i+1})$ is not, modify the value in the array corresponding to the current new edge in the same time step. To quickly detect all such $a$, for each heavy node maintain a linked list, and observe that creating a new soft link $W_a(u)$ requires adding a new element to the list of $u$, and removing one element to the list of its parent.
\item Update the array corresponding to the upper part of the edge by iterating over all letters.
\end{enumerate}

We execute this procedure incrementally over a number of updates to the light subtree. Meanwhile, the queries are answered by checking both the old and the new edge. If the only thing happening were the updates to exactly one light subtree hanging off some $v_i$, we could schedule the procedure over $\Omega(s)$ of them. Unfortunately, more than one subtree might grow too large at the same (or similar) time, and furthermore new nodes can appear on the edge. In other words, new branching nodes can appear, and the length of the edge can increase, forcing us to add new virtual branching nodes. Instead of running a separate process for each new branching node, we create one master process per edge and store a queue of nodes which should be made branching. The master process runs over $\frac{s}{4}$ updates (each of them either increasing the weight of some subtree, or creating a new node on the edge), and starts as soon as we notice that the queue is nonempty or the length of the edge exceeds $s$. Then it makes a local copy of the current queue (clearing it afterwards), and splits the edge so that all nodes in the local queue are made branching and all new edges consist of at most $\frac{3}{2}s$ nodes. To achieve the latter, if the edge was of length at least $s$ before the update, we insert one new virtual branching node in the middle, and then even if up to $\frac{s}{2}$ new nodes appear, the length of each new edge is still below $\frac{3}{2}s$, but at least $\frac{s}{2}$. As a result we get a bunch of smaller edges with potentially nonempty queues, and we might need to run their master processes as soon as the first update to their parts of the tree happens. It is clear that this guarantees that the edges never become too long, but we must also verify that any node $v_i$ is made branching after at most $\frac{s}{2}$ updates to its light subtree, and this is why we have chosen the master process to run over $\frac{s}{4}$ updates. At most $\frac{s}{4}$ first updates to the light subtree are spent waiting for an already running process to finish, and then $v_i$ belongs to some queue, so the next $\frac{s}{4}$ updates will be used to make $v_i$ (and possibly some other nodes) branching.

This proves Thm.~\ref{thm:main3}.

\bibliographystyle{abbrv}
\bibliography{paper}

\end{document}